\documentclass[12pt,draftcls,onecolumn]{IEEEtran}

\usepackage{fullpage,amssymb,amsmath,graphicx,url}

\def\Inner#1#2{ {\left\langle #1,#2\right\rangle}}
\def\inner#1#2{ {\langle #1,#2\rangle}}

\def\vector#1#2{\left[\begin{array}{cc}#1\cr #2\end{array}\right]}

\def\cqfd{\hfill\hbox{$\hbox{\vrule width 0.8pt
\vbox to6pt{\hrule depth 0.8pt width 5.2pt
\vfill\hrule depth 0.8pt}\vrule width 0.8pt}$}}

\def\JS{\mathrm{JS}}

\newcommand{\es}{\mbox{\textbf{E}}}

\newcommand{\tja}{\mathrm{tJ}_\alpha}
\newcommand{\tjopta}{\mathrm{tJ}_{\mathrm{opt},\alpha}}

\newtheorem{theorem}{Theorem}
\newtheorem{lemma}{Lemma}
\newtheorem{corollary}{Corollary}

\newtheorem{open}{Open problem}
\newtheorem{remark}{Remark}

\newenvironment{proof}[1][Proof]{\begin{trivlist}
\item[\hskip \labelsep {\bfseries #1}]}{\end{trivlist}}

\def\tJ{\mathrm{tJ}}
\def\tB{\mathrm{tB}}
\def\tr{\mathrm{tr}}
\def\KL{\mathrm{KL}}

\def\X{\mathcal{X}}

\def\SL{\mathrm{SL}}
\def\eps{\epsilon}
\def\tx{{\tilde{x}}}
\def\X{\mathcal{X}}
\def\dnu{\mathrm{d}\nu}

\author{
\authorblockN{Frank Nielsen,~\IEEEmembership{Senior Member,~IEEE}}\\
\authorblockA{Sony Computer Science Laboratories, Inc.\\
3-14-13 Higashi Gotanda, 141-0022 Shinagawa-ku\\ Tokyo, Japan\\
{\tt Frank.Nielsen@acm.org}\\
{\tt www.informationgeometry.org}
}\\
\and
\authorblockN{Richard Nock,~\IEEEmembership{Non-member}}\\
\authorblockA{CEREGMIA\\
University of Antilles-Guyane\\
Martinique, France\\
{\tt e-mail: rnock@martinique.univ-ag.fr}
}
}
 
\title{Total Jensen divergences: Definition, Properties and $k$-Means++ Clustering}

\begin{document}

\maketitle


\maketitle 

\IEEEdisplaynotcompsoctitleabstractindextext
\IEEEpeerreviewmaketitle

\begin{abstract}
We present a novel class of divergences induced by a smooth convex function called {\em total Jensen divergences}.
Those total Jensen divergences are invariant by construction to rotations, a feature yielding regularization of ordinary Jensen divergences by a conformal factor.
We analyze the relationships between this novel class of total Jensen divergences and the recently introduced total Bregman divergences.
We then proceed by defining the {\em total Jensen centroids} as average distortion minimizers, and study their robustness performance to outliers.
Finally, we prove that the $k$-means++ initialization that bypasses explicit centroid computations is good enough in practice to guarantee probabilistically a constant approximation factor to the optimal $k$-means clustering.
\end{abstract}


\section{Introduction: Geometrically designed divergences from convex functions}

A divergence $D(p:q)$ is a smooth distortion measure that quantifies the dissimilarity between any two data points $p$ and $q$ (with $D(p:q)=0$ iff. $p=q$).
A statistical divergence is a divergence between probability (or positive) measures.
There is a rich literature of divergences~\cite{Basseville-2013} that have been proposed and used depending on their axiomatic characterizations or their empirical performance. 
One motivation to design new divergence families, like the proposed total Jensen divergences in this work, is to elicit some statistical robustness property that allows to bypass the use of costly $M$-estimators~\cite{DTMRI-2011}.
We first recall the main geometric constructions of divergences from the graph plot of a convex function.
That is, we concisely review the Jensen divergences, the Bregman divergences and the total Bregman divergences induced from the graph plot of a convex generator.
We dub those divergences: {\em geometrically designed divergences} (not to be confused with the geometric divergence~\cite{GeometricDivergence-2012} in information geometry).

\subsection{Skew Jensen and Bregman divergences\label{sec:jd}}

For a strictly convex and differentiable function $F$, called the {\em generator}, we define a family of parameterized distortion measures by:

\begin{eqnarray}
J'_\alpha(p:q) &=& \alpha F(p)+(1-\alpha)F(q) -F(\alpha p+(1-\alpha)q),\quad \alpha\not\in\{0,1\},\\
 &=&  (F(p)F(q))_\alpha - F((pq)_\alpha),
\end{eqnarray}
where $(pq)_\gamma=\gamma p+(1-\gamma)q=q + \gamma(p-q) $ and $(F(p)F(q))_\gamma=\gamma F(p)+(1-\gamma)F(q)=F(q) + \gamma(F(p)-F(q))$.
The skew Jensen divergence is depicted graphically by a Jensen convexity gap, as illustrated in Figure~\ref{fig:gd} and in Figure~\ref{fig:sj}   by the plot of the generator function.
The skew Jensen divergences are {\em asymmetric} (when $\alpha\not=\frac{1}{2}$) and does {\em not} satisfy the triangular inequality of metrics.
For $\alpha=\frac{1}{2}$, we get a symmetric divergence $J'(p,q)=J'(q,p)$, also called Burbea-Rao divergence~\cite{2011-brbhat}.
It follows from the strict convexity of the generator that $J'_\alpha(p:q)\geq 0$ with equality if and only if $p=q$ (a property termed the {\em identity of indiscernibles}). The skew Jensen divergences {\em may not} be convex divergences.\footnote{A Jensen divergence is convex if and only if $\nabla^2 F(x)  \geq \frac{1}{2} \nabla^2((x+y)/2)$ for all $x,y\in\mathcal{X}$. See~\cite{ISVD:2010} for further details.} 
Note that the generator may be defined up to a constant $c$, and that $J'_{\alpha,\lambda F+c}(p:q)=\lambda J'_{\alpha,F}(p:q)$ for $\lambda>0$. 
By rescaling those divergences by a fixed factor $\frac{1}{\alpha(1-\alpha)}$, we obtain a continuous $1$-parameter family of divergences,
called the {\em $\alpha$-skew Jensen divergences},
 defined over the {\em full} real line $\alpha\in\mathbb{R}$ as follows~\cite{zhang-duality-2004,2011-brbhat}:

\begin{equation}
J_\alpha(p:q) = \left\{ 
\begin{array}{ll}
\frac{1}{\alpha(1-\alpha)} J'_\alpha(p:q) & \alpha\not =\{0,1\},\\
B(p:q) & \alpha=0,\\
B(q:p) & \alpha=1.
\end{array}
\right.
\end{equation}
where $B(\cdot:\cdot)$ denotes\footnote{For sake of simplicity, when it is clear from the context, we drop the function generator $F$ underscript in all the $B$ and $J$ divergence notations.} the Bregman divergence~\cite{bregmankmeans-2005}:

\begin{equation}
B(p:q) = F(p)-F(q)-\inner{p-q}{\nabla F(q)},
\end{equation}
with  $\inner{x}{y}=x^\top y$ denoting the inner product ({\it e.g.} scalar product for vectors).
Figure~\ref{fig:gd} shows graphically the Bregman divergence as the ordinal distance between $F(p)$ and the first-order Taylor expansion of $F$ at $q$ evaluated at $p$. 
Indeed, the limit cases of Jensen divergences $J_\alpha(p:q)=\frac{1}{\alpha(1-\alpha)} J'_\alpha(p:q)$ when $\alpha=0$ or $\alpha=1$ tend to a Bregman divergence~\cite{2011-brbhat}: \begin{eqnarray}
\lim_{\alpha\rightarrow 0} J_\alpha(p:q)&=&B(p:q),\\
\lim_{\alpha\rightarrow 1} J_\alpha(p:q)&=&B(q:p).
\end{eqnarray} 

\def\Bhat{\mathrm{Bhat}}
\def\dx{\mathrm{dx}}

The skew Jensen divergences are related to {\em statistical divergences} between probability distributions:
\begin{remark}
The statistical skew Bhattacharrya divergence~\cite{2011-brbhat}: 
\begin{equation}
\Bhat(p_1:p_2)=-\log \int p_1(x)^{\alpha} p_2(x)^{1-\alpha} \dnu(x),
\end{equation} 
between parameterized distributions
 $p_1=p(x|\theta_1)$ and $p_2=p(x|\theta_2)$ belonging to the {\em same} exponential family amounts to compute equivalently a skew Jensen divergence
 on the corresponding natural parameters for the log-normalized function $F$: $\Bhat(p_1:p_2)=J'_\alpha(\theta_1:\theta_2)$. $\nu$ is the counting measure for   discrete distributions and the Lebesgue measure for continous distributions.
See~\cite{2011-brbhat} for details.
\end{remark}

\begin{remark}
When considering a family of divergences parameterized by a smooth convex generator $F$, observe that the convexity of the function $F$ is preserved by an affine transformation. Indeed, let $G(x)=F(ax)+b$, then $G''(x)=a^2 F''(ax)>0$ since $F''>0$.
\end{remark}

\begin{figure}
\centering
\includegraphics[width=0.7\textwidth]{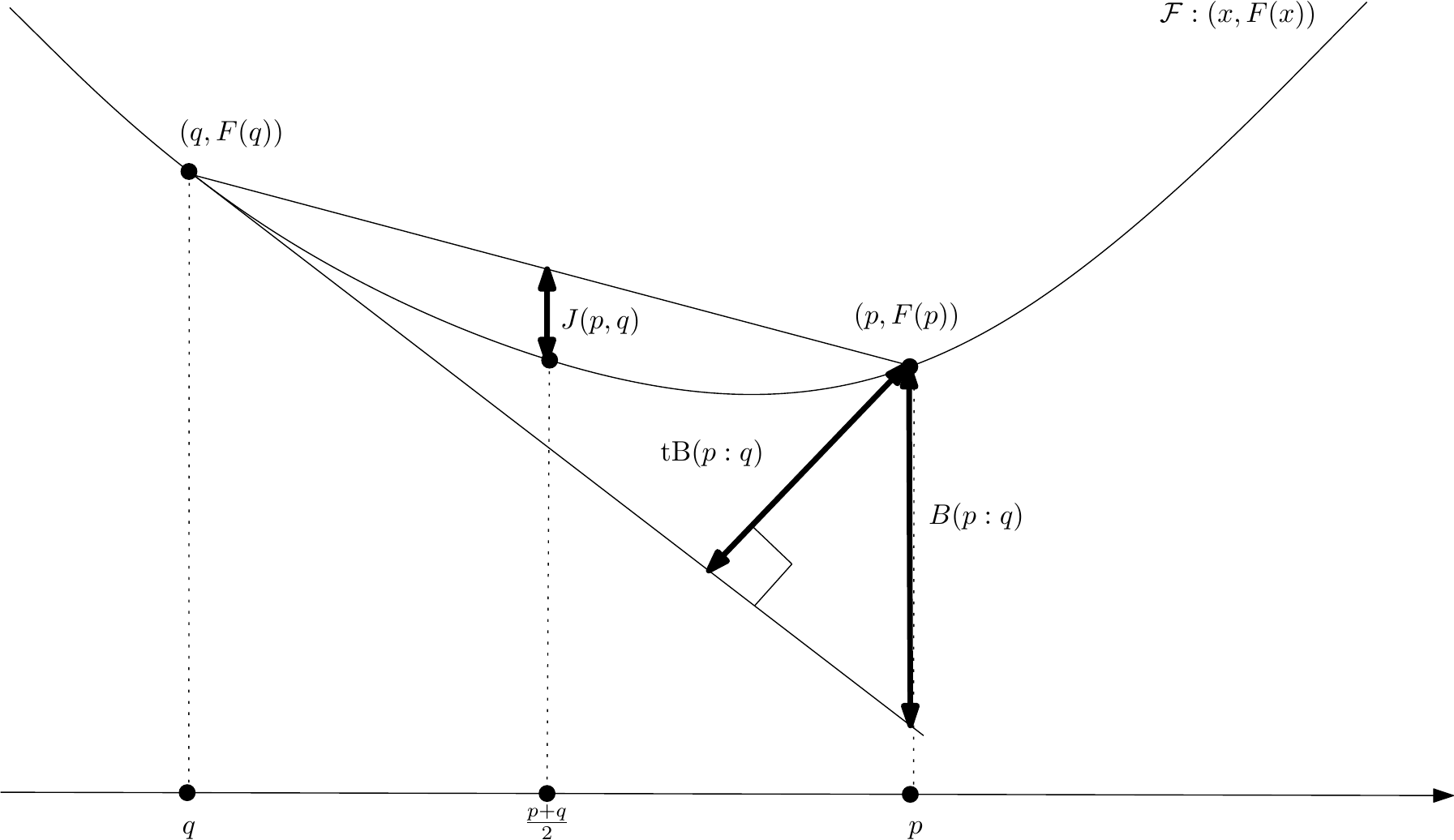}

\caption{Geometric construction of a divergence from the graph plot $\mathcal{F}=\{(x,F(x))\ |\ x\in\mathcal{X}\}$ of a convex generator $F$: The symmetric Jensen (J) divergence $J(p,q)=J_{\frac{1}{2}}(p:q)$ is interpreted as the vertical gap between the point $(\frac{p+q}{2},F(\frac{p+q}{2}))$ of $\mathcal{F}$ and the interpolated point $(\frac{p+q}{2},\frac{F(p)+F(q)}{2})$.
The asymmetric Bregman divergence (B) is interpreted as the ordinal difference between $F(p)$ and the linear approximation of $F$ at $q$ (first-order Taylor expansion) evaluated at $p$. The total Bregman (tB) divergence projects orthogonally $(p,F(p))$ onto the tangent hyperplane of $\mathcal{F}$ at $q$. 
\label{fig:gd}}
\end{figure}

\subsection{Total Bregman divergences: Definition and invariance property\label{sec:tbd}}
Let us consider an application in medical imaging to motivate the need for a particular kind of invariance when defining divergences:
In Diffusion Tensor Magnetic Resonance Imaging (DT-MRI), 3D raw data are captured at voxel positions as 3D ellipsoids denoting the water propagation characteristics~\cite{DTMRI-2011}. To perform common signal processing tasks like denoising, interpolation or segmentation tasks, one needs to define a proper {\em dissimilarity measure} between any two such ellipsoids. Those ellipsoids are mathematically handled as {\em Symmetric Positive Definite}  (SPD) matrices~\cite{DTMRI-2011} that can also be interpreted as 3D Gaussian probability distributions.\footnote{It is well-known that in information geometry~\cite{informationgeometry-2000}, the statistical $f$-divergences, including the prominent Kullback-Leibler divergence, remain unchanged after applying an invertible transformation on the parameter space of the distributions, see Appendix~\ref{app:invariance} for such an illustrating example.}
In order not to be biased by the chosen coordinate system for defining those ellipsoids, we ask for a divergence that is invariant to rotations of the coordinate system.
For a divergence parameterized by a generator function $F$ derived from the graph of that generator,  the invariance under rotations
  means that the geometric quantity defining the divergence should not change if the original coordinate system is rotated. 
  This is clearly not the case for the skew Jensen divergences that rely on the vertical axis to measure the ordinal distance, as illustrated in Figure~\ref{fig:sj}. 
To cope with this drawback, the family of {\em total Bregman divergences} (tB) have been introduced and shown to be statistically robust~\cite{DTMRI-2011}.
Note that although the traditional Kullback-Leibler divergence (or its symmetrizations like the Jensen-Shannon divergence or the Jeffreys divergence~\cite{2011-brbhat}) between two multivariate Gaussians could have been used to provide the desired invariance (see Appendix~\ref{app:invariance}), the processing tasks are not robust to outliers and perform less well in practice~\cite{DTMRI-2011}.  

The {\em total Bregman divergence} amounts to compute a scaled Bregman divergence: 
Namely a Bregman divergence multiplied by a {\em conformal factor}\footnote{In Riemannian geometry, a model geometry is said conformal if it preserves the angles. For example, in hyperbolic geometry, the Poincar\'e disk model is conformal but not the Klein disk model~\cite{hyperbolicvoronoi:2010}.
In a conformal model of a geometry, the Riemannian metric can be written as a product of the tensor metric times a conformal factor $\rho$: $g_M(p)=\rho(p) g(p)$. By analogy, we shall call a conformal divergence $D'(p:q)$, a divergence $D(p:q)$ that is multiplied by a conformal factor: $D'(p:q)=\rho(p,q) D(p:q)$.} $\rho_B$:

\begin{eqnarray}
\tB(p:q) &=& \frac{B(p:q)}{\sqrt{1+\inner{\nabla F(q)}{\nabla F(q)}}} = \rho_B(q) B(p:q),\\
\rho_B(q) &=& \frac{1}{\sqrt{1+\inner{\nabla F(q)}{\nabla F(q)}}}.\label{eq:rhob}
\end{eqnarray}
Figure~\ref{fig:gd} illustrates the total Bregman divergence as the orthogonal projection of point $(p,F(p))$ onto the tangent line of $\mathcal{F}$ at $q$.

For example, choosing the generator $F(x)=\frac{1}{2}\inner{x}{x}$ with $x\in\X=\mathbb{R}^d$, we get the {\em total square Euclidean distance}:

\begin{equation}
tE(p,q) = \frac{1}{2}\frac{\inner{p-q}{p-q}}{\sqrt{1+\inner{q}{q}}}.
\end{equation}
That is, $\rho_B(q)=\sqrt{\frac{1}{1+\inner{q}{q}}}$ and $B(p:q)=\frac{1}{2}\inner{p-q}{p-q}=\frac{1}{2} \|p-q\|^2_2$.

Total Bregman divergences~\cite{PhDMeizhu-2011}  have proven successful in many applications: Diffusion tensor imaging~\cite{DTMRI-2011} (DTI), 
shape retrieval~\cite{2012-tBDShapeRetrieval}, boosting~\cite{tBD-Boosting-2011}, multiple object tracking~\cite{tBD-tracking-2013}, 
tensor-based graph matching~\cite{tBD-Tensor-2011}, just to name a few.

The total Bregman divergences can be defined over the space of symmetric positive definite (SPD) matrices met in DT-MRI~\cite{DTMRI-2011}.
One key feature of the total Bregman divergence defined over such matrices is its invariance under the special linear  group~\cite{PhDMeizhu-2011} $\SL(d)$ that consists of $d \times d$ matrices of unit determinant:

\begin{equation}
\tB(A^\top P A: A^\top Q A) =   \tB(P:Q),\ \forall A\in\SL(d).
\end{equation}
See also appendix~\ref{app:invariance}.

\subsection{Outline of the paper}
The paper is organized as follows: Section~\ref{sec:tJ} introduces the novel class of total Jensen divergences (tJ) and present some of its properties and relationships with the total Bregman divergences. In particular, we show that although the square root of the Jensen-Shannon is a metric, it is not true anymore for the square root of the total Jensen-Shannon divergence.
Section~\ref{sec:centroid} defines the centroid with respect to total Jensen divergences, and proposes a doubly iterative algorithmic scheme to approximate those centroids.
The notion of robustness of centroids is studied via the framework of the influence function of an outlier.
Section~\ref{sec:cluster} extends the $k$-means++ seeding to total Jensen divergences, yielding a fast probabilistically guaranteed clustering algorithm that do not require to compute explicitly those centroids. Finally, Section~\ref{sec:concl} summarizes the contributions and concludes this paper.

\section{Total Jensen divergences\label{sec:tJ}}

\subsection{A geometric definition}
Recall that the skew Jensen divergence $J_\alpha'$ is defined as the ``vertical'' distance between the interpolated point $((pq)_\alpha, (F(p)F(q))_\alpha)$ lying on the line segment $[(p,F(p)),(q,F(q))]$ and the point $((pq)_\alpha, F((pq)_\alpha)$ lying on the graph of the generator.
This measure is therefore {\em dependent} on the coordinate system chosen for representing the sample space $\mathcal{X}$ since the notion of ``verticality'' depends on the coordinate system.
To overcome this limitation, we define the {\em total} Jensen divergence by choosing the {\em unique orthogonal projection} of $((pq)_\alpha, F((pq)_\alpha)$ onto the line 
$((p,F(p)),(q,F(q))])$. 
Figure~\ref{fig:sj} depicts graphically the total Jensen divergence, and shows its invariance under a rotation.
Figure~\ref{fig:projoutside} illustrates the fact that the orthogonal projection of point $((pq)_\alpha, F((pq)_\alpha)$ onto the line $((p,F(p)),(q,F(q)))$ may fall outside the segment $[(p,F(p)),(q,F(q))]$.

\begin{figure}
\centering
\includegraphics[width=\textwidth]{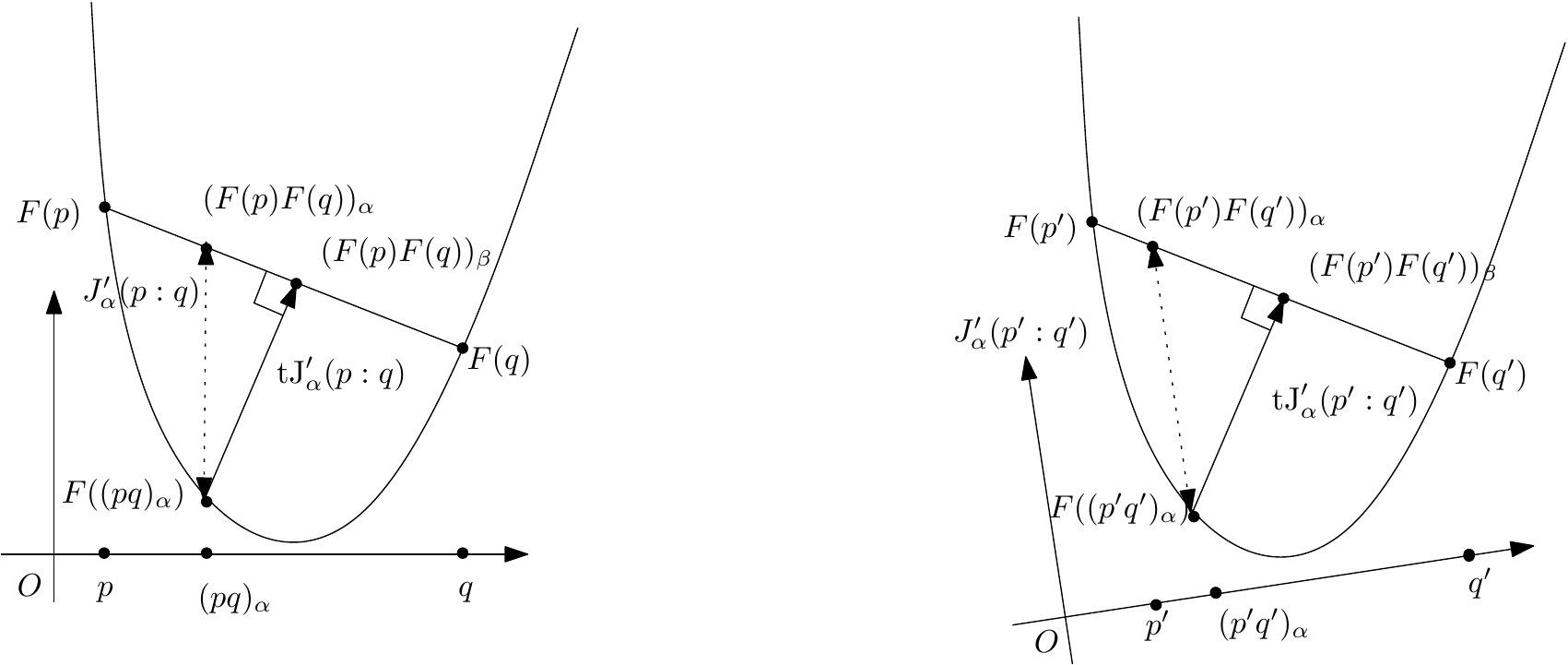}

\caption{The total Jensen $\tJ'_\alpha$ divergence is defined as the unique orthogonal projection of point $((pq)_\alpha, F((pq)_\alpha))$ onto
the line passing through $(p,F(p))$ and $(q,F(q))$.
It is invariant by construction to a rotation of the coordinate system. (Observe that the bold line segments have same length but the dotted line segments have different lengths.)
Except for the paraboloid $F(x)=\frac{1}{2}\inner{x}{x}$, the divergence is not invariant to translations.
\label{fig:sj}}
\end{figure}

We report below a geometric proof (an alternative purely analytic proof is also given in the Appendix~\ref{app:analytic}).
Let us plot the epigraph of function $F$ restricted to the vertical plane passing through distinct points $p$ and $q$.
Let $\Delta_F=F(q)-F(p)\in\mathbb{R}$ and $\Delta=q-p\in\mathcal{X}$  (for $p\not =q$).
Consider the two right-angle triangles $\Delta T_1$ and $\Delta T_2$ depicted in Figure~\ref{fig:gproof}.
Since Jensen divergence $J$ and  $\Delta_F$ are vertical line segments intersecting 
the line passing through point $(p,F(p))$ and point $(q,F(q))$, we deduce that the angles $\hat{b}$ are the same.
Thus it follows that angles $\hat{a}$ are also identical.
Now, the cosine of angle $\hat{a}$ measures the ratio of the adjacent side over the hypotenuse of right-angle triangle $\triangle T_2$.
 Therefore it follows that: 

\begin{equation}
\cos \hat{a} = \frac{\|\Delta\|}{\sqrt{\inner{\Delta}{\Delta}+\Delta_F^2}}=\sqrt{\frac{1}{1+\frac{\Delta_F^2}{\inner{\Delta}{\Delta}}}},
\end{equation}
where $\|\cdot\|$ denotes the $L_2$-norm.
In right-triangle $\triangle T_1$, we furthermore deduce that:

\begin{equation}
\tJ_\alpha'(p:q) = J_\alpha'(p:q) \cos\hat{a} = \rho_J(p,q) J_\alpha'(p:q).
\end{equation}

Scaling by factor $\frac{1}{\alpha(1-\alpha)}$, we end up with the following lemma:

\begin{lemma}
The total Jensen divergence $\tJ_\alpha$ is invariant to rotations of the coordinate system.
The divergence is mathematically expressed as a scaled skew Jensen divergence 
$\tJ_\alpha(p:q) = \rho_J(p,q) J_\alpha(p:q)$, where $\rho_J(p,q)=\sqrt{\frac{1}{1+\frac{\Delta_F^2}{\inner{\Delta}{\Delta}}}}$ is symmetric and independent of the skew factor.
\end{lemma}

Observe that the scaling factor $\rho_J(p,q)$ is {\em independent} of $\alpha$, symmetric, and is always less or equal to $1$. 
Furthermore, observe that the scaling factor depending on both $p$ and $q$ and is {\em not separable}: That is, $\rho_J$ cannot be expressed as a product of two terms, one depending only on $p$ and the other depending only on $q$.
That is, $\rho_J(p,q)\not=\rho'(p)\rho'(q)$.
We have $\tJ_{1-\alpha}(p:q)=\rho_J(p,q) J_{1-\alpha}(p:q)=\rho_J(p,q)J_\alpha(q:p)$.
Because the conformal factor is independent of $\alpha$, we have the following asymmetric ratio equality:

\begin{equation}
\frac{\tJ_\alpha(p:q)}{\tJ_\alpha(q:p)}=\frac{J_\alpha(p:q)}{J_\alpha(q:p)}.
\end{equation}

By rewriting,

\begin{equation}
\rho_J(p,q)  = \sqrt{\frac{1}{1+s^2}},
\end{equation}
we intepret the non-separable {\em conformal factor} as a function of the square of the {\em chord slope} $s=\frac{\Delta_F}{\|\Delta\|}$.
Table~\ref{tab:tj} reports the conformal factors for the total Bregman and total Jensen divergences.

\begin{figure}
\centering
\includegraphics[width=0.7\textwidth]{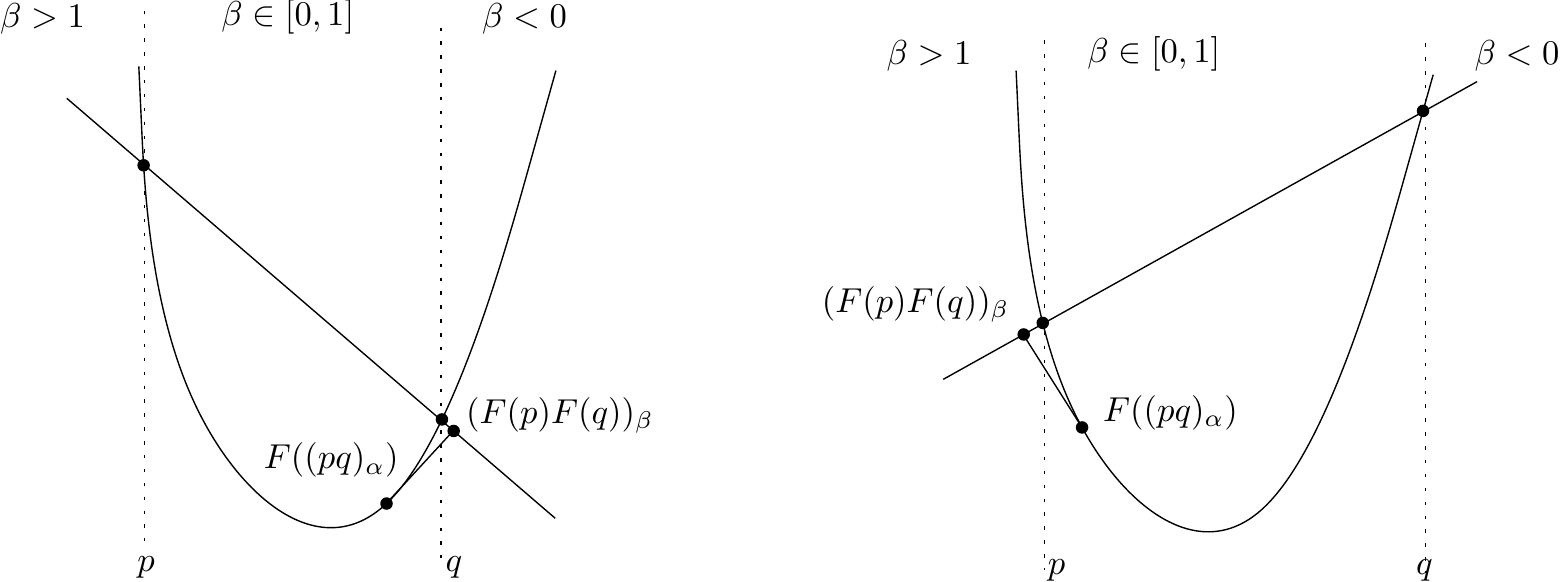}

\caption{The orthogonal projection of the point $((pq)_\alpha, F((pq)_\alpha)$ onto the line $((p,F(p)),(q,F(q)))$ at position $((pq)_\beta,(F(p)F(q))_\beta)$ may fall outside the line segment $[(p,F(p)),(q,F(q))]$. 
Left: The case where  $\beta<01$. Right: The case where $\beta>1$. (Using the convention $(F(p)F(q))_\beta=\beta F(p)+(1-\beta)F(q)$.)
\label{fig:projoutside}}
\end{figure}

\begin{table}
\centering

\begin{eqnarray*}
\tB(p:q) &=& \rho_B(q) B(p:q), \quad \rho_B(q)=\sqrt{\frac{1}{1+\inner{\nabla F(q)}{\nabla F(q)}}} \\
\tJ_\alpha(p:q) &=& \rho_J(p,q) J_\alpha(p:q),\quad \rho_J(p,q)=\sqrt{\frac{1}{1+\frac{(F(p)-F(q))^2}{\inner{p-q}{p-q}}}}\\
\end{eqnarray*}

$$
\begin{array}{|l|l|l|l|l|}\hline
\text{Name (information)}& \text{Domain $\X$} & \text{Generator $F(x)$} & \rho_B(q) & \rho_J(p,q)\\ \hline\hline
\text{Shannon} & \mathbb{R}^+ & x\log x-x  & \sqrt{\frac{1}{1+\log^2 q}} & \sqrt{\frac{1}{1+\left(\frac{p\log p+q\log q+p-q}{p-q}\right)^2}}\\
\text{Burg} & \mathbb{R}^+ & -\log x & \sqrt{\frac{1}{1+\frac{1}{q^2}}} &   \sqrt{\frac{1}{1+\left(\frac{\log \frac{q}{p}}{p-q}\right)^2}} \\
\text{Bit} & [0,1] & x\log x+(1-x)\log (1-x) &  \sqrt{\frac{1}{1+\log^2 \frac{q}{1-q}}} &  \sqrt{\frac{1}{1+\left(\frac{p\log p+(1-p)\log p -q\log q-(1-q)\log(1-q)}{p-q}\right)^2}} \\
\text{Squared Mahalanobis} & \mathbb{R}^d & \frac{1}{2} x^\top Q x & \sqrt{\frac{1}{1+\|Q q\|_2^2}} &   
\sqrt{\frac{1}{1+\left(\frac{p Q^\top p - q Q^\top q}{\|p-q\|_2}\right)^2}},\quad  Q\succ 0\\ \hline
\end{array}
$$
\caption{Examples of conformal factors for the total Bregman divergences ($\rho_B$) and the total skew Jensen divergences ($\rho_J$).
\label{tab:tj}}
\end{table}

\begin{remark}\label{rq:differenttJ}
We could have defined the total Jensen divergence in a different way by fixing the point on the line segment $((pq)_\alpha,(F(p)F(q))_\alpha) \in [(p,F(p)),(q,F(q))]$ and seeking the point $((pq)_\beta,F((pq)_\beta))$ on the function plot that ensures that the line segments
$[(p,F(p)),(q,F(q))]$ and $[((pq)_\beta,F((pq)_\beta)),((pq)_\alpha,(F(p)F(q))_\alpha)]$ are orthogonal. However, this yields a mathematical condition that may not always have a closed-form solution for solving for $\beta$. This explains why the other way around (fixing the point on the function plot and seeking the point on the line $[(p,F(p)),(q,F(q))]$ that yields orthogonality) is a better choice since we have a simple closed-form expression.
See Appendix~\ref{app:secondtotaldiv}.
\end{remark}

From a scalar divergence $D_F(x)$ induced by a generator $F(x)$, we can always build a {\em separable divergence} in $d$ dimensions by adding coordinate-wise the axis scalar divergences: 
\begin{equation}
D_F(x)=\sum_{i=1}^d D_F(x_i).
\end{equation}
This divergence corresponds to the induced divergence obtained by considering  the multivariate separable generator:
\begin{equation}
F(x)=\sum_{i=1}^d F(x_i).
\end{equation}

The Jensen-Shannon divergence~\cite{KDiv-LinWong-1990,LinJS-1991} is a separable Jensen divergence for the Shannon information generator $F(x)=x\log x -x$:

\begin{eqnarray}
\JS(p,q) &=& \frac{1}{2} \sum_{i=1}^d p_i\log \frac{2p_i}{p_i+q_i} + \frac{1}{2} \sum_{i=1}^d q_i\log \frac{2q_i}{p_i+q_i},\\
&=& J_{F(x)=\sum_{i=1}^d x_i\log x_i,\alpha=\frac{1}{2}}(p:q).
\end{eqnarray}

Although the Jensen-Shannon divergence is symmetric it is not a metric since it fails the triangular inequality.
However, its square root $\sqrt{\JS(p,q)}$ is a metric~\cite{FugledeTopsoe:2004}.
One may ask whether the conformal factor $\rho_{\mathrm{JS}}(p,q)$ destroys or preserves the metric property:

\begin{lemma}
The square root of the total Jensen-Shannon divergence is not a metric.
\end{lemma}

\def\tJS{\mathrm{tJS}}

It suffices to report a counterexample as follows:
Consider the three points of the $1$-probability simplex
 $p=(0.98,0.02)$, $q=(0.52,0.48)$ and $r=(0.006,0.994)$.
We have $d_1=\sqrt{\tJS(p,q)}\simeq 0.35128346734040883$,
$d_2=\sqrt{\tJS(q,r)}\simeq 0.39644485899866616$ and 
$d_3=\sqrt{\tJS(p,r)}\simeq  0.7906141593521927$.
The triangular inequality fails because $d_1+d_2<d_3$.
The triangular inequality deficiency is $d_3-(d_1+d_2) \simeq 0.04288583301311766$.

\subsection{Total Jensen divergences: Relationships with total Bregman divergences}

Although the underlying rationale for deriving the total Jensen divergences followed the same principle of the total Bregman divergences ({\em i.e.}, replacing the ``vertical'' projection by an orthogonal projection), the total Jensen divergences {\em do not coincide} with the total Bregman divergences in limit cases:
Indeed, in the limit cases $\alpha\in\{0,1\}$, we have:

\begin{eqnarray}
\lim_{\alpha\rightarrow 0} \tJ_\alpha(p:q) &=& \rho_J(p,q) B(p:q)\not = \rho_B(q) B(p:q) , \\
\lim_{\alpha\rightarrow 1} \tJ_\alpha(p:q) &=& \rho_J(p,q) B(q:p)\not = \rho_B(p) B(q:p),
\end{eqnarray}
since $\rho_J(p,q)\not =\rho_B(q)$.
Thus when $p\not=q$, the total Jensen divergence {\em does not tend} in limit cases to the total Bregman divergence.
However, by using a Taylor expansion  with exact Lagrange remainder, we write:

\begin{equation}
F(q)=F(p)+\inner{q-p}{\nabla F(\epsilon)},
\end{equation}
with $\epsilon\in[p,q]$ (assuming wlog. $p<q$).
That is, $\Delta_F =  F(q)-F(p) = \inner{\nabla F(\epsilon)}{\Delta}$.
Thus we have the squared slope index:
\begin{equation}
s^2 = \frac{\Delta_F^2}{\|\Delta\|^2} = \frac{\Delta^\top \nabla F(\eps) \Delta^\top \nabla F(\eps)}{\Delta^\top \Delta}= \inner{\nabla F(\eps)}{\nabla F(\eps)} =\|\nabla F(\eps)\|^2.
\end{equation}

Therefore when $p\simeq q$, we have $\rho_J(p,q) \simeq \rho_B(q)$, and the total Jensen divergence tends to the total Bregman divergence for any value of $\alpha$.
Indeed, in that case, the Bregman/Jensen conformal factors match:
\begin{equation}
\rho_J(p,q) = \frac{1}{ \sqrt{1+\inner{\nabla F(\epsilon)}{\nabla F(\epsilon)}}  }  = \rho_B(\epsilon),
\end{equation}
for $\epsilon\in[p,q]$. Note that for univariate generators, we find explicitly the value of $\eps$:
\begin{equation} 
\eps = \nabla F^{-1}\left(\frac{\Delta_F}{\Delta}\right) = \nabla F^{*}\left(\frac{\Delta_F}{\Delta}\right),
\end{equation}
where $F^*$ is the Legendre convex conjugate, see~\cite{2011-brbhat}. 
We recognize the expression of a Stolarsky mean~\cite{StolarskyMean-1975} of $p$ and $q$ for the strictly monotonous function $\nabla F$.

Therefore when $p\sim q$, we have  $\lim_{p\sim q}\frac{\Delta_F}{\Delta}\rightarrow \nabla F(q)$ and the total Jensen divergence
converges to the total Bregman divergence. 

\begin{lemma}
The total skew Jensen divergence $\tJ_\alpha(p:q)=\rho_J(p,q) J(p:q)$ can be equivalently rewritten as $\tJ_\alpha(p:q)=\rho_B(\eps) J(p:q)$ for $\eps\in [p,q]$.
In particular, when $p\simeq q$, the total Jensen divergences tend to the total Bregman divergences for any $\alpha$.
For $\alpha\in\{0,1\}$, the total Jensen divergences  is not equivalent to the total Bregman divergence when $p\not\simeq q$.
\end{lemma}

\begin{remark}
Let the chord slope $s(p,q)=\frac{F(q)-F(p)}{\|q-p\|}= \|\nabla F(\eps)\|$ for $\eps\in[pq]$ (by the mean value theorem).
It follows that $s(p,q)\in [\min_\eps \|\nabla F(\eps)\|, \max_\eps \|\nabla F(\eps)\|]$.
Note that since $F$ is strictly convex, $\nabla F$ is strictly increasing and we can approximate arbitrarily finely using a simple 1D dichotomic search the value of $\eps$:
$\eps_0=p+\lambda(q-p)$ with $\lambda\in[0,1]$ and $\frac{F(q)-F(p)}{\|q-p\|} \sim \|\nabla F(\eps_0)\|$.
In 1D, it follows from the strict increasing monotonicity of $F'$ that $s(p,q)\in [|F'(p)|, |F'(q)|]$ (for $p\leq q$).
\end{remark}

\begin{remark}
Consider the chord slope $s(p,q)=\frac{F(q)-F(p)}{\|q-p\|}$ with one fixed extremity, say $p$.
When $q\rightarrow \pm\infty$, using l'Hospital rule, provided that $\lim \nabla F(q)$ is bounded, we have 
$\lim_{q\rightarrow\infty} s(p,q) = \nabla F(q)$.
Therefore, in that case, $\rho_J(p,q)\simeq \rho_B(q)$. Thus when $q\rightarrow\pm\infty$, we have $\tJ_\alpha(p:q)\rightarrow \rho_B(q) J_\alpha(p:q)$. In particular, when $\alpha\rightarrow 0$ or $\alpha\rightarrow 1$, the total Jensen divergence tends to the total Bregman divergence. 
Note that when $q\rightarrow \pm\infty$, we have $J_\alpha(p:q)\simeq \frac{1}{\alpha(1-\alpha)} ((1-\alpha)F(q)-F((1-\alpha)q))$.
\end{remark}

We finally close this section by noticing that total Jensen divergences may be interpreted as a kind of Jensen divergence with the generator function scaled by the conformal factor:

\begin{lemma}\label{lemma:tJequalJ}
The total Jensen divergence $\tJ_{F,\alpha}(p:q)$ is equivalent to a Jensen divergence for the convex generator $G(x)=\rho_J(p,q)F(x)$:
$\tJ_{F,\alpha}(p:q) = J_{\rho_J(p,q)F,\alpha}(p:q)$.
\end{lemma}

%
%
%
%

\begin{figure}
\centering

\includegraphics[width=0.7\textwidth]{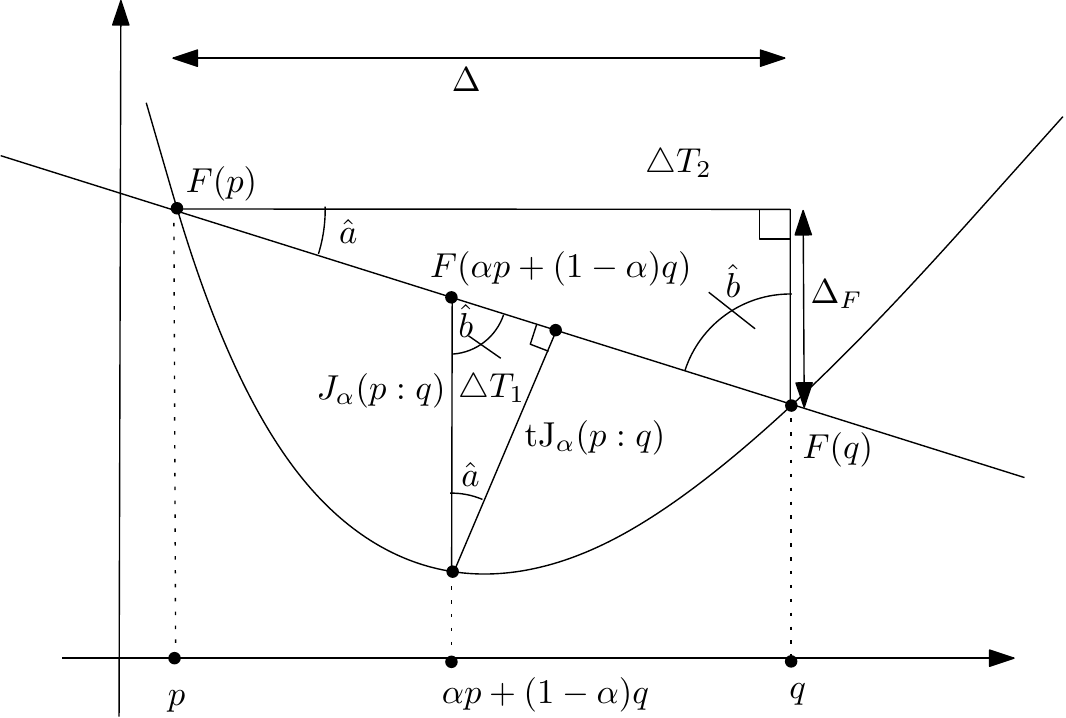}

\caption{Geometric proof for the total Jensen divergence: The figure illustrates the two right-angle triangles $\triangle T_1$ and $\triangle T_2$. We deduce that angles $\hat{b}$ and $\hat{a}$ are respectively identical from which the formula on $\mathrm{tJ}_\alpha(p:q)=J_\alpha(p:q)\cos\hat{a}$ follows immediately. \label{fig:gproof}}

\end{figure}

\section{Centroids and robustness analysis\label{sec:centroid}}

\subsection{Total Jensen centroids}
Thanks to the invariance to rotations, total Bregman divergences proved highly useful in applications (see~\cite{DTMRI-2011,2012-tBDShapeRetrieval,tBD-Boosting-2011,tBD-tracking-2013,tBD-Tensor-2011}, etc.) due to the statistical robustness of their centroids. The conformal factors play the role of regularizers.
Robustness of the centroid, defined as a notion of centrality robust to ``noisy'' perturbations, is studied using the framework of the {\em influence function}~\cite{Hampel-IF-1986}.
Similarly, the total skew Jensen (right-sided) {\em centroid} $c_\alpha$ is defined for a finite weighted point set as the minimizer of the following loss function:

\begin{eqnarray}
L(x;w) &=& \sum_{i=1}^n  w_i \times \tJ_\alpha(p_i:x),\\
c_\alpha &=& \arg\min_{x\in\mathcal{X}} L(x;w) , \label{eq:centroid}
\end{eqnarray}
where $w_i\geq 0$ are the normalized point weights (with $\sum_{i=1}^n w_i=1$).
The left-sided centroids $c_\alpha'$ are obtained by minimizing the equivalent right-sided centroids for $\alpha'=1-\alpha$: $c_\alpha'=c_{1-\alpha}$ (recall that the conformal factor does not depend on $\alpha$). Therefore, we consider the right-sided centroids in the remainder.

To minimize Eq.~\ref{eq:centroid}, we proceed iteratively in two stages: 
\begin{itemize}

\item First, we consider $c^{(t)}$ (initialized with the barycenter $c^{(0)}=\sum_i w_ip_i$)  given. 
This allows us to consider the following simpler 
minimization problem:

\begin{equation}\label{eq:centroidfixed}
c = \arg\min_{x\in\mathcal{X}} \sum_{i=1}^n  w_i \times \rho(p_i,c^{(t)})  J_\alpha(p_i:x).
\end{equation}

Let 
\begin{equation}
w_i^{(t)}= \frac{ w_i \times \rho(p_i,c^{(t)})}{\sum_j w_j \times \rho(p_j,c^{(t)})},
\end{equation}
be the updated renormalized weights at stage $t$.

\item Second, we   minimize:

\begin{equation}\label{eq:centroidfixedweight}
c = \arg\min_{x\in\mathcal{X}} \sum_{i=1}^n   w_i^{(t)}  J_\alpha(p_i:x).
\end{equation}
This is a convex-concave minimization procedure~\cite{cccp-2003} (CCCP) that can be solved iteratively until it reaches convergence~\cite{2011-brbhat} at $c^{(t+1)}$. That is, we iterate the following formula~\cite{2011-brbhat} a given number of times $k$:

\begin{equation}
c^{t,l+1} \leftarrow (\nabla F)^{-1} \left( \sum_{i=1}^n w_i^{(t)} \nabla F(\alpha c^{t,l}+(1-\alpha) p_i) \right)
\end{equation}

We set the new centroid $c^{(t+1)}=c^{t,k}$, and we loop back to the first stage until the loss function improvement $L(x;w)$ goes below a prescribed threshold. (An implementation in Java\texttrademark{} is available upon request.)
\end{itemize}

Although the CCCP algorithm is guaranteed to converge {\em monotonically} to a local optimum, the two steps weight update/CCCP does not provide anymore of the  monotonous converge as we have attested in practice.

It is an open and challenging problem to prove that the total Jensen centroids are unique whatever the chosen multivariate generator, see~\cite{2011-brbhat}.

In Section~\ref{sec:cluster}, we shall show that a careful initialization by picking centers from the dataset is enough to guarantee a probabilistic bound on the $k$-means clustering with respect to total Jensen divergences {\em without} computing the centroids. 

\subsection{Centroids: Robustness analysis\label{sec:robustness}}

The centroid defined with respect to the total Bregman divergence has been shown to be robust to  outliers whatever the chosen generator~\cite{DTMRI-2011}.
We first analyze the robustness for the symmetric Jensen divergence (for $\alpha=\frac{1}{2}$).
We investigate the {\em influence function}~\cite{Hampel-IF-1986} $i(y)$ on the centroid when adding an outlier point $y$ with prescribed weight $\epsilon>0$.
Without loss of generality, it is enough to consider two points: One {\em outlier} with $\epsilon$ mass and one {\em inlier} with the remaining mass.
Let us add an outlier point $y$ with weight $\epsilon$ onto an inliner point  $p$.
Let $\bar{x}=p$ and $\tilde{x}=p+\epsilon z$ denote the centroids before adding $y$ and after adding $y$. 
$z=z(y)$ denotes the influence function.
The Jensen centroid minimizes (we can ignore dividing by the renormalizing total weight inlier+outlier: $\frac{1}{1+\epsilon}$):

$$
L(x) \equiv  J(p,x) + \eps J(x,y).
$$

The derivative of this energy is:

$$
D(x) = L'(x) = J'(p,x) + \eps J'(y,x).
$$

The derivative of the Jensen divergence is given by (not necessarily a convex distance):
$$
J'(h,x)=\frac{1}{2} f'(x)-\frac{1}{2} f'\left(\frac{x+h}{2}\right),
$$
where $f$ is the univariate convex generator and $f'$ its derivative.
For the optimal value of the centroid $\tx$, we have $D(\tx)=0$, yielding:

\begin{equation}
(1+\eps)  f'(\tx) - \left(f'\left(\frac{\tx+p}{2}\right) + \eps f'\left(\frac{\tx+y}{2}\right) \right) =0.
\end{equation}

Using Taylor expansions on $\tx=p+\eps z$  (where $z=z(y)$ is the influence function) on the derivative $f'$, we get :

$$
f'(\tx) \simeq f'(p)+\eps z f''(p),
$$
and

$$
(1+\eps)(f'(p)+\eps z f''(p)) - \left(f'(p)+\frac{1}{2}\eps z f''(p) + \eps f'\left(\frac{p+y}{2}\right)\right)
$$
(ignoring the term in $\eps^2$ for small constant $\eps>0$ in the Taylor expansion term of $\eps f'$.)

Thus we get the following mathematical equality:

$$
z((1+\eps)\eps f''(p) -\eps z/2 f''(p)) = f'(p)+\eps f'\left(\frac{p+y}{2}\right) - (1+\eps) f'(p)
$$

Finally, we get the expression of the influence function:

\begin{equation}
z = z(y) = 2 \frac{f'(\frac{p+y}{2})-f'(p)}{f''(p)},
\end{equation}
for small prescribed $\epsilon>0$.

\begin{theorem}
The Jensen centroid is robust for a strictly convex and smooth generator $f$ if 
$|f'(\frac{p+y}{2})|$ is bounded on the  domain $\X$ for any prescribed $p$.
\end{theorem}
Note that this theorem extends to separable divergences. 
To illustrate this theorem, let us consider two   generators  that shall yield a non-robust and a robust Jensen centroid, respectively:

\begin{itemize}

\item Jensen-Shannon: $\X=\mathbb{R}^+$, $f(x)=x\log x-x$ ,$f'(x)=\log(x)$, $f''(x)=1/x$.

We check that $|f'(\frac{p+y}{2})| = |\log \frac{p+y}{2}| $ is unbounded when $y\rightarrow +\infty$. 
The influence function $z(y) = 2 p\log \frac{p+y}{2p}$ is unbounded when $y\rightarrow \infty$, and therefore the centroid is not robust to outliers.

\item Jensen-Burg: $\X=\mathbb{R}^+$, $f(x)=-\log x$, $f'(x)=-1/x$, $f'' (x)=\frac{1}{x^2}$

We check that $|f'(\frac{p+y}{2})| = | \frac{2}{p+y} |$ is always bounded for $y\in (0,+\infty)$. 

$$
z(y) = 2 p^2 \left( \frac{1}{p} - \frac{2}{p+y}\right)
$$

When $y\rightarrow \infty$, we have $|z(y)| \rightarrow 2p<\infty$.
The influence function is bounded and the centroid is robust.
 
\end{itemize}

\begin{theorem}
Jensen centroids are not always robust (e.g., Jensen-Shannon centroids).
\end{theorem}

Consider the total Jensen-Shannon centroid. Table~\ref{tab:tj} reports the total Bregman/total Jensen conformal factors.
Namely, $\rho_B(q)=\sqrt{\frac{1}{1+\log^2 q}}$ and
$\rho_J(p,q)=\sqrt{\frac{1}{1+\left(\frac{p\log p+q\log q+p-q}{p-q}\right)^2}}$.
Notice that the total Bregman centroids  have been proven to be robust (the left-sided $t$-center in~\cite{DTMRI-2011}) whatever the chosen generator.

For the total Jensen-Shannon centroid (with $\alpha=\frac{1}{2}$), the conformal factor of the outlier point tends to:

\begin{equation}
\lim_{y\rightarrow +\infty} \rho_J(p,y) \simeq \frac{1}{\log y}. 
\end{equation}

It is an ongoing investigation to prove (1) the uniqueness of Jensen centroids, and (2) the robustness of total Jensen centroids whathever the chosen generator (see the conclusion).


\subsection{Clustering: No closed-form centroid, no cry!\label{sec:cluster}}

The most famous clustering algorithm is $k$-means~\cite{Lloyd-1957} that consists in first initializing $k$ distinct seeds (set to be the initial cluster centers) and then iteratively assign the points to their closest center, and update the cluster centers by taking the centroids of the clusters.
A breakthrough was achieved by proving the a randomized seed selection, $k$-means++~\cite{kmeansplusplus-2007}, 
guarantees probabilistically a constant approximation factor to the optimal loss.
The $k$-means++ initialization may be interpreted as a discrete $k$-means where the $k$ cluster centers are choosen among the input.
This yields ${n \choose k}$ combinatorial seed sets. 
Note that $k$-means is NP-hard when $k=2$ and the dimension is not fixed, but not discrete $k$-means~\cite{Kmeans2D-NPHard-2013}.
Thus we do not need to compute centroids to cluster with respect to total Jensen divergences. 
Skew Jensen centroids can be approximated arbitrarily finely using the concave-convex procedure, as reported in~\cite{2011-brbhat}.

\begin{lemma}
On a compact domain $\cal{X}$, we have 
$\rho_{\min} J(p:q) \leq  \tJ(p:q) \leq  \rho_{\max} J(p:q)$,
with $\rho_{\min} = \min_{x\in\X} \frac{1}{\sqrt{1+\inner{\nabla F(x)}{\nabla F(x)}}}$
and $\rho_{\max} = \max_{x\in\X} \frac{1}{\sqrt{1+\inner{\nabla F(x)}{\nabla F(x)}}}$
\end{lemma}

We are given a set ${\mathcal{S}}$ of points that we wish to cluster
in $k$ clusters, following a hard clustering assignment. We let $\tja(A:y) = \sum_{x \in A} {\tja(x:y)}$ for
any $A \subseteq S$. The optimal total hard clustering Jensen
potential is $\tja^{\mathrm{opt}} = \min_{C \subseteq S : |C| =
k} \tja(C)$, where $\tja(C) = \sum_{x \in S} \min_{c \in
  C} \tja(x:c)$. Finally,
the contribution of some $A \subseteq S$ to the optimal total Jensen potential
having centers $C$ is $\tjopta(A) = \sum_{x\in A} \min_{c \in
  C} \tja(x:c)$. 

Total Jensen seeding picks randomly without replacement an element $x$
in $S$ with probability proportional to $\tja(C)$, where $C$ is the
current set of centers. When $C=\emptyset$, the distribution is
uniform. We let total Jensen seeding refer to this biased randomized
way to pick centers.

We state the theorem that applies for any kind of divergence:

\begin{theorem}\label{th1}
Suppose there exist some $U$ and $V$ such that, $\forall
x, y, z$:
\begin{eqnarray}
\tja(x:z) & \leq & U(\tja(x:y) + \tja(y:z))\:\:, \mbox{ (triangular
  inequality)} \label{ti}\\
\tja(x:z) & \leq & V\tja(z:x) \:\:, \mbox{ (symmetric
  inequality)} \label{si}
\end{eqnarray}
Then the average potential of total Jensen seeding with $k$ clusters satisfies $E[\tja] \leq
2 U^2(1+V) (2+\log k)\tjopta$, where $\tjopta$ is the minimal total
Jensen potential achieved by a clustering in $k$ clusters.
\end{theorem}

The proof is reported in Appendix~\ref{app:triangularineq}.

To find values for $U$ and $V$, we make two assumptions, denoted H, on $F$, that are
supposed to hold in the convex closure of ${\mathcal{S}}$ (implicit
in the assumptions):
\begin{itemize}
\item First, the maximal condition number of the Hessian of $F$, that is, the ratio between
the maximal and minimal eigenvalue ($>0$) of the Hessian of $F$, is
upperbounded by $K_1$. 

\item Second, we
assume the Lipschitz condition on $F$ that $\Delta_F^2 /
\inner{\Delta}{\Delta} \leq K_2$, for some $K_2>0$.
\end{itemize}

Then we have the following lemma:

\begin{lemma}\label{lem0}
Assume $0<\alpha<1$. Then, under assumption H, for any $p, q, r \in {\mathcal{S}}$, there
exists $\epsilon>0$ such that:
\begin{eqnarray}
\tJ_\alpha(p:r) & \leq & \frac{2 (1+K_2) K_1^2}{\epsilon}\left(\frac{1}{1-\alpha} \tJ_\alpha(p:q) + \frac{1}{\alpha} \tJ_\alpha(q:r)\right)\:\:.\label{triang1}
\end{eqnarray}
\end{lemma}
Notice that because of Eq.~\ref{ppp}, the right-hand side of Eq.~\ref{triang1} tends to zero when $\alpha$ tends to $0$ or $1$. A consequence of Lemma \ref{lem0} is the following triangular inequality:

\begin{corollary}
The total skew Jensen divergence satisfies the following triangular inequality:
\begin{eqnarray}
\tJ_\alpha(p:r) & \leq & \frac{2 (1+K_2) K_1^2}{\epsilon\alpha(1-\alpha)}\left(\tJ_\alpha(p:q) + \tJ_\alpha(q:r)\right)\:\:.\label{triang2}
\end{eqnarray}
\end{corollary}
The proof is reported in Appendix~\ref{app:triangularineq}.
Thus, we may set $U=\frac{2 (1+K_2) K_1^2}{\epsilon}$ in Theorem~\ref{th1}.

\begin{lemma}
Symmetric inequality condition (\ref{si}) holds for $V = K_1^2(1+K_2)/\epsilon$, for some $0<\epsilon<1$.
\end{lemma}
The proof is reported in Appendix~\ref{app:triangularineq}. 

\section{Conclusion\label{sec:concl}}

We described a novel family of divergences, called total skew Jensen divergences, that are invariant by rotations.
Those divergences scale the ordinary Jensen divergences by a conformal factor independent of the skew parameter, and extend naturally the underlying principle of the former total Bregman divergences~\cite{2012-tBDShapeRetrieval}. 
However, total Bregman divergences differ from total Jensen divergences in limit cases because of the non-separable property of the conformal factor $\rho_J$ of total Jensen divergences. Moreover, although the square-root of the Jensen-Shannon divergence (a Jensen divergence) is a metric~\cite{FugledeTopsoe:2004}, this does
not hold anymore for the total Jensen-Shannon divergence.
Although the total Jensen centroids do not admit a closed-form solution in general, we proved that the convenient $k$-means++ initialization probabilistically guarantees a constant approximation factor to the optimal clustering.
We also reported a simple two-stage iterative algorithm to approximate those total Jensen centroids.
We studied the robustness of those total Jensen and Jensen centroids and proved that skew Jensen centroids robustness depend on the considered generator.

In information geometry, conformal divergences have been recently investigated~\cite{escort-flat-2010}: The single-sided conformal transformations (similar to total Bregman divergences) are shown to flatten the $\alpha$-geometry of the space of discrete distributions into a dually flat structure. Conformal factors have also been introduced in
machine learning~\cite{ConformalKernelSVM-2002} to improve the classification performance of Support Vector Machines (SVMs). In that latter case, the conformal transfortion considered is double-sided separable: $\rho(p,q)=\rho(p)\rho(q)$.
It is interesting to consider the underlying geometry of total Jensen divergences that exhibit non-separable double-sided conformal factors.
A Java\texttrademark{} program that implements the total Jensen centroids are available upon request.

To conclude this work, we leave two open problems to settle:

\begin{open}
Are (total) skew Jensen centroids defined as minimizers of weighted distortion averages unique? 
Recall that Jensen divergences may not be convex~\cite{ISVD:2010}. The difficult case is then when $F$ is multivariate non-separable~\cite{2011-brbhat}. 
\end{open}

\begin{open}
Are total skew Jensen centroids always robust? 
That is, is the influence function of an outlier always bounded?
\end{open}

 

\appendix

\section*{An analytic proof for  the total Jensen divergence formula\label{app:analytic}}
Consider Figure~\ref{fig:sj}.
To define the total skew Jensen divergence, we write the {\em orthogonality constraint} as follows:

\begin{equation}
\Inner{ \vector{p-q}{F(p)-F(q)}}{ \vector{(pq)_\alpha - (pq)_\beta }{F((pq)_\alpha)- (F(p)F(q))_\beta} } =0.
\end{equation}
Note that $(pq)_\alpha - (pq)_\beta=(\alpha-\beta)(p-q)$.
Fixing one of the two unknowns\footnote{We can fix either $\alpha$ or $\beta$, and solve for the other quantity. 
Fixing $\alpha$ always yield a closed-form solution. Fixing $\beta$ may not necessarily yields a closed-form solution. Therefore, we consider $\alpha$ prescribed in the following.}, we solve for the other quantity, and define the total Jensen divergence by
the Euclidean distance between the two points:

\begin{eqnarray}
\tJ'_\alpha(p:q) &=& \sqrt{ (p-q)^2 (\alpha-\beta)^2  + ((F(p)F(q))_\beta-F((pq)_\alpha)^2 },\\
&=& \sqrt{ \Delta^2 (\alpha-\beta)^2  +  (F(q)+\beta \Delta_F  - F(q+\alpha \Delta)) ^2 }
\end{eqnarray}

For a prescribed $\alpha$,  let us solve for $\beta\in\mathbb{R}$:

\begin{eqnarray}
\beta &=& \frac{(F(p)-F(q))(F( (pq)_\alpha )-F(q))+(p-q)( (pq)_\alpha-q)}{(p-q)^2 + (F(p)-F(q))^2 },\\
&=&\frac{(F(p)-F(q))(F( (pq)_\alpha )-F(q))+\alpha(p-q)^2}{(p-q)^2 + (F(p)-F(q))^2 },\\
&=&\frac{\Delta_F (F(q+\alpha\Delta)-F(q))+\alpha\Delta^2}{\Delta^2 + \Delta_F^2 }
\end{eqnarray}

Note that for $\alpha\in(0,1)$, $\beta$ may be falling outside the unit range $[0,1]$.
It follows that:
\begin{eqnarray}
\alpha-\beta(\alpha) &=& \frac{\alpha\Delta^2+\alpha\Delta_F^2-\Delta_F(F(q+\alpha\Delta)-F(q))-\alpha\Delta^2}{\Delta^2+\Delta_F^2}\\
&=& \frac{\Delta_F(\alpha F(p)+(1-\alpha)F(q)-F(\alpha p+(1-\alpha) q))}{\Delta^2+\Delta_F^2}\\
&=& \frac{\Delta_F }{\Delta^2+\Delta_F^2} J'_\alpha(p:q)
\end{eqnarray} 

We apply Pythagoras theorem from  the right triangle in Figure~\ref{fig:sj}, and get: 

\begin{equation}
\underbrace{ ( (pq)_\alpha-(pq)_\beta )^2+( (F(p)F(q))_\alpha - (F(p)F(q))_\beta )^2}_{l=(\alpha-\beta)^2(\Delta^2+\Delta_F^2)}  +  \tJ'^2(p:q)=J'^2(p:q)
\end{equation}
That is, the length $l$ is $|\alpha-\beta|$ times the length of the segment linking $(p, F(p))$ to $(q, F(q))$.
Thus we have:
\begin{equation}
l = \frac{|\Delta_F|}{\sqrt{\Delta^2+\Delta_F^2}}  J'_\alpha(p:q)
\end{equation}
Finally, get the (unscaled) total Jensen divergence:
\begin{eqnarray}
\tJ_\alpha'(p:q) &=&\sqrt{\frac{\inner{\Delta}{\Delta}}{\inner{\Delta}{\Delta}+\Delta_F^2}} J_\alpha'(p:q),\\
&=& \rho_J(p,q) J_\alpha'(p:q).
\end{eqnarray}

\section*{A different definition of total Jensen divergences\label{app:secondtotaldiv}}
As mentioned in Remark~\ref{rq:differenttJ}, we could have used the orthogonal projection principle by fixing the point $((pq)_\beta,(F(p)F(q))_\beta)$ on the line segment $[(p,F(p)),(q,F(q))]$ and seeking for a point $((pq)_\alpha,F((pq)_\alpha))$ on the function graph that yields orthogonality. However, this second approach does not always give rise to a closed-form solution for the total Jensen divergences.
Indeed, fix $\beta$ and consider the point $((pq)_\beta,(F(p)F(q))_\beta))$ that meets the orthogonality constraint.
Let us solve for $\alpha\in[0,1]$, defining the point $((pq)_\alpha,F((pq)_\alpha))$ on the graph function plot.
Then we have to solve for the unique $\alpha\in[0,1]$ such that:

\begin{equation}
\Delta_F F((pq)_\alpha)+\alpha\Delta^2= \beta (\Delta^2+\Delta_F^2)+ \Delta_F F(q).
\end{equation}

In general, we cannot explicitly solve this equation although we can always approximate finely the solution.
Indeed, this equation amounts to solve:
\begin{equation}
b F(q+\alpha (p-q)) +\alpha c -a =0,
\end{equation}
with coefficients $a=\beta (\Delta^2+\Delta_F^2)+\Delta_F F(q)$, $b=\Delta_F$, and $c=\Delta^2$.
In general, it does not admit a closed-form solution because of the $F(q+\alpha (p-q))$ term.
For the squared Euclidean potential $F(x)=\frac{1}{2}\inner{x}{x}$, or the separable Burg entropy $F(x)=-\log x$, we obtain a closed-form solution.
However, for the most interesting case, $F(x)=x\log x-x$ the Shannon information (negative entropy), we cannot get a closed-form solution for $\alpha$.

The length of the segment joining the points $((pq)_\alpha,F((pq)_\alpha))$ and $((pq)_\beta,(F(p)F(q))_\beta)$ is:

\begin{equation}
\sqrt{(\alpha-\beta)^2\Delta^2 + (F((pq)_\alpha)-(F(p)F(q))_\beta)^2}.
\end{equation}

Scaling by $\frac{1}{\beta(1-\beta)}$, we get the second kind of total Jensen divergence:

\begin{equation}
\tJ_\beta^{(2)}(p:q) = \frac{1}{\beta(1-\beta)} \sqrt{(\alpha-\beta)^2\Delta^2 + (F((pq)_\alpha)-(F(p)F(q))_\beta)^2}.
\end{equation}

The second type total Jensen divergences do  not tend to total Bregman divergences when $\alpha\rightarrow 0,1$.

\section*{Information Divergence and Statistical Invariance\label{app:invariance}}

Information geometry~\cite{informationgeometry-2000} is primarily concerned with the differential-geometric modeling of the space of statistical distributions. In particular, statistical manifolds induced by a parametric family of distributions have been historically of prime importance.
In that case, the statistical divergence $D(\cdot:\cdot)$ defined on such a probability manifold $\mathcal{M}=\{p(x|\theta)\ |\ \theta\in\Theta\}$ between two random variables $X_1$ and $X_2$ (identified using their corresponding parametric probability density functions $p(x|\theta_1)$ and $p(x|\theta_2)$) should be {\em invariant} by (1) a {\em one-to-one mapping} $\tau$ of the parameter space, and by (2) sufficient statistics~\cite{informationgeometry-2000} (also called a Markov morphism). 
The first condition means that $D(\theta_1:\theta_2) = D(\tau(\theta_1):\tau(\theta_2)), \forall \theta_1,\theta_2\in\Theta$.
As we illustrate next, a transformation on the observation space $\mathcal{X}$ (support of the distributions) may sometimes induce an {\em equivalent transformation} on the parameter space $\Theta$.
In that case, this yields as a byproduct a third invariance by the transformation on $\mathcal{X}$.

For example, consider a $d$-variate normal random variable $x\sim N(\mu,\Sigma)$ and apply a rigid transformation $(R,t)$ to $x\in\mathcal{X}=\mathbb{R}^d$ so that we get $y=Rx+t\in\mathcal{Y}=\mathbb{R}^d$.
Then $y$ follows a normal distribution.
In fact, this normal$\leftrightarrow$normal conversion is well-known to hold for any {\em invertible affine transformation} $y=A(x)=Lx+t$ with $\det(L) \not =0$ on $\X$, see~\cite{MGMM-2001}: $y \sim N(L\mu + t, L \Sigma L^\top)$.

Consider the Kullback-Leibler (KL) divergence (belonging to the $f$-divergences) 
between $x_1\sim N(\mu_1,\Sigma_1)$ and $x_2\sim N(\mu_2,\Sigma_2)$, and let us show that this amounts to calculate equivalently the KL divergence between $y_1$ and $y_2$: 

\begin{equation}
\KL(x_1:x_2)=\KL(y_1:y_2).
\end{equation}

The KL divergence $\KL(p:q)=\int p(x) \log\frac{p(x)}{q(x)} \dx$ between two $d$-variate normal distributions $x_1\sim N(\mu_1,\Sigma_1)$ and $x_2\sim N(\mu_2,\Sigma_2)$ is given by:
$$
\KL(x_1:x_2)=\frac{1}{2} \left( \tr(\Sigma_2^{-1}\Sigma_1) + \Delta\mu^\top \Sigma_2^{-1} \Delta\mu - \log \frac{|\Sigma_1|}{|\Sigma_2|} -d\right),
$$
with $\Delta\mu=\mu_1-\mu_2$ and $|\Sigma|$ denoting the determinant of the matrix $\Sigma$.

Consider $\KL(y_1:y_2)$ and let us show that the formula is equivalent to $\KL(x_1:x_2)$.
We have the term $R\Sigma_2^{-1} R^\top R \Sigma_1 R^\top=R\Sigma_2^{-1}  \Sigma_1 R^\top$, and from the trace cyclic property, we rewrite this expression as
$\tr(R\Sigma_2^{-1} \Sigma_1  R^\top)=\tr(\Sigma_2^{-1} \Sigma_1 R^\top R)=\tr(\Sigma_2^{-1} \Sigma_1)$ since $R^\top R=I$.
Furthermore, $\Delta\mu'=R(\mu_2-\mu_1)=R\Delta\mu$.
It follows that:
\begin{equation}
\Delta\mu_y^\top R \Sigma_2^{-1} R^\top \Delta\mu_y =  \Delta\mu^\top R^\top R \Sigma_2^{-1} R^\top R\Delta\mu=\Delta\mu^\top \Sigma_2^{-1} \Delta\mu,
\end{equation}
and
\begin{equation}
 \log \frac{|R \Sigma_1 R^\top|}{|R \Sigma_2 R^\top|}=\log \frac{|R| |\Sigma_1| |R^\top|}{|R| |\Sigma_2| |R^\top|} = \log \frac{|\Sigma_1|}{|\Sigma_2|}.
\end{equation}

Thus we checked that the KL of Gaussians is indeed preserved by a rigid transformation in the observation space $\mathcal{X}$: $\KL(x_1:x_2)=\KL(y_1:y_2)$.

\section*{Proof: A $k$-means++ bound for total Jensen divergences\label{app:triangularineq}}
 
\begin{proof}
The key to proving the Theorem is the following Lemma, which parallels
Lemmata from \cite{avKM}
(Lemma 3.2), \cite{nlkMB} (Lemma 3).

\begin{lemma}
Let $A$ be an arbitrary cluster of $C_{opt}$ and $C$ an arbitrary
clustering. 
If we add a random element $x$ from $A$ to $C$, then 
\begin{eqnarray}
E_{x\sim \pi_S} [\tja(A, x) | x\in A] & \leq & 2 U^2(1+V) \tjopta(A)\:\:,\label{eqp1}
\end{eqnarray}
where $\pi_S$ denotes the seeding distribution used to sample $x$.
\end{lemma}
\begin{proof}
First, $E_{x\sim \pi_S} [\tja(A, x) | x\in A] = E_{x\sim \pi_A}
[\tja(A, x)]$ because we require $x$ to be in $A$. We then have:
\begin{eqnarray}
\tja(y:c_y) & \leq & \tja(y:c_x)\\
& \leq & U (\tja(y:x) + \tja(x:c_x))
\end{eqnarray}
and so averaging over all $x\in {\mathcal{A}}$ yields:
\begin{eqnarray}
\tja(y:c_y) & \leq & U \left(\frac{1}{|{\mathcal{A}}|}\sum_{x\in {\mathcal{A}}}{\tja(y:x)} + \frac{1}{|{\mathcal{A}}|}\sum_{x\in {\mathcal{A}}}{\tja(x:c_x)}\right)
\end{eqnarray}
The participation of ${\mathcal{A}}$ to the potential is:
\begin{eqnarray}
\es_{{y} \sim{} {\pi}_{\mathcal{A}}}
[\tja({\mathcal{A}},{y})] & = & \sum_{{y} \in {\mathcal{A}}} {\left\{\frac{\tja(y:c_y)}{\sum_{{x} \in {\mathcal{A}}} {\tja(x:c_x)}} \sum_{{x} \in {\mathcal{A}}} {\min\left\{\tja(x:c_x), \tja(x:y)\right\}} \right\}}\:\:. \label{firstp}
\end{eqnarray}
We get:
\begin{eqnarray}
\es_{{y} \sim{} {\pi}_{\mathcal{A}}}
[\tja({\mathcal{A}},{y})] & \leq &
\frac{U}{|{\mathcal{A}}|}\sum_{{y} \in {\mathcal{A}}}
{\left\{\frac{\sum_{x\in {\mathcal{A}}}{\tja(y:x)}}{\sum_{{x} \in
        {\mathcal{A}}} {\tja(x:c_x)}} \sum_{{x} \in {\mathcal{A}}}
    {\min\left\{\tja(x:c_x), \tja(x:y)\right\}} \right\}}\nonumber\\
  & & +  \frac{U}{|{\mathcal{A}}|}\sum_{{y} \in {\mathcal{A}}}
{\left\{\frac{\sum_{x\in {\mathcal{A}}}{\tja(x:c_x)}}{\sum_{{x} \in
        {\mathcal{A}}} {\tja(x:c_x)}} \sum_{{x} \in {\mathcal{A}}}
    {\min\left\{\tja(x:c_x), \tja(x:y)\right\}} \right\}} \label{eqq2}\\
 & \leq & \frac{2U}{|{\mathcal{A}}|} \sum_{{x} \in
   {\mathcal{A}}}{\sum_{{y} \in {\mathcal{A}}}{\tja(x:y)}} \label{eqq3}\\
 & \leq & 2U^2 \left( \sum_{{x} \in
   {\mathcal{A}}}{\tja(x:c_x)} + \frac{1}{|{\mathcal{A}}|}\sum_{{x} \in
   {\mathcal{A}}}{\sum_{{y} \in {\mathcal{A}}}{\tja(c_x:y)}}\right) \label{eqq4}\\
 & & = 2U^2 \left( \tjopta({\mathcal{A}}) + \frac{1}{|{\mathcal{A}}|}\sum_{{x} \in
   {\mathcal{A}}}{\sum_{{y} \in {\mathcal{A}}}{\tja(c_x:y)}}\right) \label{eqq5}\\
& \leq & 2U^2 \left( \tjopta({\mathcal{A}}) + \frac{V}{|{\mathcal{A}}|}\sum_{{x} \in
   {\mathcal{A}}}{\sum_{{y} \in
     {\mathcal{A}}}{\tja(y:c_x)}}\right)\label{pppp1}\\
 & & = 2U^2(1+V) \tjopta({\mathcal{A}})\:\:,\label{eqq6}
\end{eqnarray}
as claimed. We have used (\ref{ti}) in (\ref{eqq2}), the ``min''
replacement procedure of \cite{avKM}
(Lemma 3.2), \cite{nlkMB} (Lemma 3) in (\ref{eqq3}), a second time
(\ref{ti}) in (\ref{eqq4}), (\ref{si}) in (\ref{pppp1}), and the
definition of $\tjopta({\mathcal{A}})$ in (\ref{eqq5}) and (\ref{eqq6}).
\end{proof}
There remains to use exactly the same proof as in \cite{avKM}
(Lemma 3.3), \cite{nlkMB} (Lemma 4) to finish up the proof of Theorem \ref{th1}.
\end{proof}

We make two assumptions, denoted H, on $F$, that are
supposed to hold in the convex closure of ${\mathcal{S}}$ (implicit
in the assumptions):
\begin{itemize}
\item First, the maximal condition number of the Hessian of $F$, that is, the ratio between
the maximal and minimal eigenvalue ($>0$) of the Hessian of $F$, is
upperbounded by $K_1$. 

\item Second, we
assume the Lipschitz condition on $F$ that $\Delta_F^2 /
\inner{\Delta}{\Delta} \leq K$, for some $K_2>0$.
\end{itemize}

\begin{lemma} 
Assume $0<\alpha<1$. Then, under assumption H, for any $p, q, r \in {\mathcal{S}}$, there
exists $\epsilon>0$ such that:
\begin{eqnarray}
\tJ_\alpha(p:r) & \leq & \frac{2 (1+K_2) K_1^2}{\epsilon}\left(\frac{1}{1-\alpha} \tJ_\alpha(p:q) + \frac{1}{\alpha} \tJ_\alpha(q:r)\right)\:\:.\label{ntriang1}
\end{eqnarray}
\end{lemma}

\begin{proof}
Because the Hessian $H_F$ of $F$ is symmetric positive definite, it
can be diagonalized as $H_F(s) = (P D P^\top)(s)$ for some unitary matrix $P$
and positive diagonal matrix $D$. 
Thus for any $x$, $\langle x, H_F(s) x\rangle
= \langle (D^{1/2} P^\top)(s) x, (D^{1/2} P^\top)(s) x\rangle = \|(D^{1/2} P^\top)(s) x\|_2^2$. 
Applying Taylor expansions to $F$
and its gradient $\nabla F$ yields that for any $p, q, r$ in ${\mathcal{S}}$ there exists
$0<\beta, \gamma, \delta, \beta', \gamma', \delta', \epsilon <1$ and $s, s',s''$ in
the simplex $pqr$ such that:

\begin{eqnarray}
\lefteqn{2(F( (pq)_\alpha ) +
F( (qr)_\alpha ) - F( (pr)_\alpha ) - F(q))}\nonumber\\
 & = & \alpha^2 (\beta'+\gamma'-1) \|(D^{1/2} P^\top)(s') (p-q)\|_2^2 + (1-\alpha)^2 (\beta+\gamma-1) \|(D^{1/2} P^\top)(s)
 (q-r)\|_2^2  \nonumber\\
 & & + \alpha(1-\alpha) \langle ((D^{1/2} P^\top)(s) + (D^{1/2} P^\top)(s')) (p-q),((D^{1/2} P^\top)(s) + (D^{1/2} P^\top)(s')) (q-r)\rangle
 \nonumber\\
 & \leq & \alpha^2 \|(D^{1/2} P^\top)(s') (p-q)\|_2^2 + (1-\alpha)^2 \|(D^{1/2} P^\top)(s)
 (q-r)\|_2^2  \nonumber\\
 & & + 2\alpha(1-\alpha) \|(D^{1/2} P^\top)(s'') (p-q)\|_2 \|(D^{1/2}
 P^\top)(s'') (q-r)\|_2 \nonumber\\
& \leq & \frac{\rho_F^2}{\epsilon \alpha(1-\alpha)}\left(\alpha^2 J_\alpha(p:q) + (1-\alpha)^2 J_\alpha(q:r) + 2\alpha(1-\alpha) \sqrt{ J_\alpha(p:q)  J_\alpha(q:r)}\right)
 \label{eqj}\\
 & \leq & \frac{2 \rho_F^2}{\epsilon}\left(\frac{\alpha}{1-\alpha} J_\alpha(p:q) + \frac{1-\alpha}{\alpha} J_\alpha(q:r)\right)\:\:. \label{eqj2}
\end{eqnarray}
Eq. (\ref{eqj}) holds because a Taylor expansion of $J_\alpha(p:q)$ yields that
$\forall p, q$,
there exists $0<\epsilon, \gamma<1$ such that 
\begin{eqnarray}
J_\alpha(p:q) & = & \alpha(1-\alpha) \epsilon \|(D^{1/2}
P^\top)((pq)_\gamma) (p-q)\|_2^2\:\:.\label{ppp}
\end{eqnarray} 
Hence, using (\ref{eqj2}), we obtain:
\begin{eqnarray}
J_\alpha(p:r) & = & J_\alpha(p:q) + J_\alpha(q:r) + F( (pq)_\alpha ) +
F( (qr)_\alpha ) - F( (pr)_\alpha ) - F(q)\nonumber\\
 & \leq & \frac{2 \rho_F^2}{\epsilon}\left(\frac{1}{1-\alpha} J_\alpha(p:q) + \frac{1}{\alpha} J_\alpha(q:r)\right)\:\:.
\end{eqnarray}
Multiplying by $\rho(p,q)$ and using the fact that $\rho(p,q) /
\rho(p',q')<1+K$, $\forall p,p',q,q'$, yields the statement of the Lemma.
\end{proof}
Notice that because of Eq.~\ref{ppp}, the right-hand side of Eq.~\ref{triang1} tends to zero when $\alpha$ tends to $0$ or $1$. A consequence of Lemma \ref{lem0} is the following triangular inequality:

\begin{corollary}
The total skew Jensen divergence satisfies the following triangular inequality:
\begin{eqnarray}
\tJ_\alpha(p:r) & \leq & \frac{2 (1+K_2) K_1^2}{\epsilon\alpha(1-\alpha)}\left(\tJ_\alpha(p:q) + \tJ_\alpha(q:r)\right)\:\:. 
\end{eqnarray}
\end{corollary}

\begin{lemma}
Symmetric inequality condition (\ref{si}) holds for $V = K_1^2(1+K_2)/\epsilon$, for some $0<\epsilon<1$.
\end{lemma}
\begin{proof}
We make use again of (\ref{ppp}) and have, for some $0<\epsilon,
\gamma, \epsilon', \gamma'<1$ and any $p, q$, both $J_\alpha(p:q) = \alpha(1-\alpha) \epsilon \|(D^{1/2}
P^\top)((pq)_\gamma) (p-q)\|_2^2$ and $J_\alpha(q:p)  = J_{1-\alpha}(p:q) = \alpha(1-\alpha) \epsilon' \|(D^{1/2}
P^\top)((pq)_{\gamma'}) (p-q)\|_2^2$, which implies, for some:
$0 < \epsilon'' <1$
\begin{eqnarray}
\frac{J_\alpha(p:q)}{J_{\alpha}(q:p)} & \leq & \frac{K_1^2}{\epsilon''} 
\end{eqnarray} 
Using the fact that $\rho(p,q) / \rho(q, p) \leq 1 + K_2$, $\forall p,
q$, yields the statement of the Lemma.
\end{proof}

\end{document}